\documentclass[runningheads]{llncs}
\usepackage[utf8]{inputenc}
\usepackage[inline]{enumitem}
\usepackage{todonotes}
\usepackage{thm-restate}
\usepackage{booktabs}
\usepackage{amsmath,amssymb}
\usepackage{subcaption}
\usepackage{wrapfig}
\usepackage{multirow}
\usepackage{lineno}
\graphicspath{{}} 

\newcommand{\prooflink}[1]{\hypersetup{linkcolor=lipicsBulletGray}\hyperref[#1]{$\blacktriangledown$}}
\newcommand{\statlink}[1]{\hypersetup{linkcolor=lipicsYellow}\hyperref[#1]{$\blacktriangle$}}
\newcommand{\mylink}{}
\newcommand{\cprooflink}[1]{\hypersetup{linkcolor=lipicsBulletGray}\hyperref[#1]{$\triangledown$}}
\newcommand{\cstatlink}[1]{\hypersetup{linkcolor=lipicsYellow}\hyperref[#1]{$\vartriangle$}}

\usepackage{xspace}
\newcommand{\OctR}{\textsc{Octilinear Realizability}\xspace}
\newcommand{\OctC}{\textsc{Octilinear Compaction}\xspace}
\newcommand{\OrtR}{\textsc{Orthogonal Realizability}\xspace}
\newcommand{\OrtC}{\textsc{Orthogonal Compaction}\xspace}

\usepackage{thm-restate}

\begin{document}
\titlerunning{Compaction and Realizability of Almost Convex Octilinear Representations}
\title{On Compaction and Realizability of Almost Convex Octilinear Representations\thanks{This work started at the Bertinoro Workshop on Graph Drawing 2024. We thank the organizers for organizing this event and the other participants for useful discussions.}} 


\author{Henry Förster\inst{1,2}\orcidID{0000-0002-1441-4189}
\and
Giacomo Ortali\inst{3}\orcidID{0000-0002-4481-698X}
\and Lena Schlipf\inst{2}\orcidID{0000-0001-7043-1867}
}

\institute{School of Computation, Information and Technology, TU Munich, Heilbronn\\
\email{henry.foerster@tum.de} \and
Wilhelm-Schickard-Institut f\"ur Informatik, Universit\"at T\"ubingen\\\email{lena.schlipf@uni-tuebingen.de} \and Department of Engineering, University of Perugia\\\email{giacomo.ortali@unipg.it}}



\authorrunning{Henry Förster, Giacomo Ortali and Lena Schlipf} 

\maketitle

\begin{abstract}
      Octilinear graph drawings are a standard  paradigm extending the orthogonal graph drawing style by two additional slopes ($\pm 1$). We are interested in two constrained drawing problems where the input specifies a so-called representation, that is: a planar embedding; the angles occurring between adjacent edges; the bends along each edge.  In \textsc{Orthogonal Realizability} one is asked to compute any orthogonal drawing satisfying the constraints, while in \textsc{Orthogonal Compaction}  the goal is to find such a drawing using minimum area.
      While \OrtR can be solved in linear time, \OrtC is NP-hard even if the graph is a cycle. In contrast, already \textsc{Octilinear Realizability}  is known to be NP-hard. 
      
      In this paper we investigate \OctR  and  \textsc{Octilinear Compaction}  problems. We prove that \OctR remains NP-hard if at most one face is not convex or if each interior face has at most $8$ reflex corners. We also strengthen the hardness proof of \OctC, showing that \OctC does not admit a PTAS even if the representation has no reflex corner except at most $4$ incident to the external face. On the positive side, we prove that \OctR is FPT in the number of reflex corners and for \OctC we describe an XP algorithm on the number of edges represented with a $\pm 1$ slope segment (i.e, the diagonals), again for the case where the representation has no reflex corner except at most $4$ incident to the external face.
\end{abstract}

\section{Introduction}
\label{se:intro}
Octilinear (graph) drawings have been the standard visualization paradigm~\cite{DBLP:journals/vlc/HongMN06,DBLP:journals/tvcg/NollenburgW11,DBLP:journals/tvcg/StottRMW11} used in metro maps since the first map of the London Underground was designed by Henry Beck in 1933~\cite{beck}\footnote{See also~\cite{beck2} for Beck's visualization.}. It extends the popular orthogonal (graph) drawing style by two additional slopes ($\pm 1$) which can be useful in various diagramming applications such as UML~\cite{uml} or BPMN~\cite{bpmn}. Similar to studies on orthogonal drawings, research efforts on octilinear drawings have aimed to compute drawings with few bends per edge in small area~\cite{DBLP:journals/jgaa/BekosG0015,DBLP:journals/jgaa/BekosKK17}.

Due to the similarities to the orthogonal drawing model, one may be tempted to adopt techniques that have been successfully employed in the context of orthogonal drawing. However, many problems related to octilinear drawing turn out to be more difficult than the corresponding ones for the orthogonal model. For instance, it is NP-hard to compute an octilinear drawing with the fewest total number of bends if the embedding is fixed~\cite{octiNP} while the corresponding problem for orthogonal drawings is polynomial time solvable~\cite{DBLP:journals/siamcomp/Tamassia87}. 

We are interested in two constrained drawing problems where the input specifies not only a planar embedding but a so-called representation, that is, a planar embedding and the angles occurring between adjacent edges, the bends along each edge, and the angles each bend forms. In the  \textsc{Realizability} problem one is then asked to compute any drawing satisfying the constraints whereas in the \textsc{Compaction} problem the goal is to find such a drawing using minimum area. \textsc{Orthogonal Realizability} can be trivially solved in polynomial time~\cite{DBLP:books/ph/BattistaETT99,DBLP:journals/siamcomp/Tamassia87}, but \textsc{Orthogonal Compaction}  is NP-hard~\cite{Patrignani01} even if the graph is a cycle~\cite{DBLP:journals/comgeo/EvansFKSSW22}.  In contrast, already  \textsc{Octilinear Realizability} is known to be NP-hard~\cite{DBLP:journals/algorithmica/BekosFK19}. The NP-hardness of  \OrtC implies the NP-hardness of \textsc{Octilinear Compaction}.

\OrtC is efficiently solvable under certain constraints~\cite{DBLP:journals/comgeo/BridgemanBDLTV00,DBLP:conf/ipco/KlauM99}. In the past few years, these concepts have resurfaced in the graph drawing community~\cite{DBLP:journals/jgaa/BekosBBDGKPR22,DBLP:conf/grapp/Esser19a} culminating in an FPT algorithm~\cite{DBLP:conf/sofsem/DidimoGKLWZ23} parameterized by the number of so-called \emph{kitty corners}. Intuitively, {kitty corners} are pairs of reflex corners that ``point'' towards each other and are a measure for non-convexity. More in general, recently various parametrized algorithms have been presented for other variants of orthogonal drawings (see for example \cite{DBLP:conf/sofsem/DidimoGKLWZ23,DBLP:conf/isaac/Didimo0LOP23,DBLP:journals/algorithmica/GiacomoDLMO24}), as far as we know no parametrized approach has been investigated for octilinear drawings.

The rest of the paper is structured as follows. In Section~\ref{sec:prelim} we give preliminary definitions.  In Section \ref{se:hard} we strengthen the hardness results known for both \OctR  and \OctC problems. In particular, in Section~\ref{sec:inapprox} we prove that \OctC  admits no $\frac{9}{4}$-approximation even when the representation has no reflex corner except at most $4$ incident to the external face; in Section \ref{sec:paranp} we show that \OctR is para-NP-hard\footnote{Para-NP-hardness means that the problem remains NP-hard for a bounded value of a parameter.} when parameterized by the  number of faces with reflex corners or by the maximum number of reflex corners per face.  In  Section~\ref{se:param} we present two parametrized algorithms: in Section \ref{sec:fpt} we show that \OctR is FPT parameterized by the number of reflex corners and in Section \ref{sec:xp} we show that \OctC admits an XP algorithm parametrized by the number of edges of the representation represented with segments of slope $\pm 1$ (i.e., the diagonals), again in the case where the representation has no reflex corner (except at most $4$ incident to the external face).

\section{Preliminaries}
\label{sec:prelim}

In an octilinear drawing of a graph $G = (V, E)$ each vertex $v \in V$ is represented by a point in $\mathbb{R}^2$ and each edge $e \in E$ is drawn as a sequence of horizontal, vertical and diagonal (with slope $\pm 1$) segments. In this paper we only investigate simple planar drawings, hence from now on we omit the word ``simple'' and the word ``planar''.  Observe that a necessary condition for a vertex to admit an octilinear drawing is having maximum degree at most 8. 

Two important equivalence classes of octilinear drawings are \emph{embeddings} and \emph{octilinear representations}. An embedding contains all (octilinear) graph drawings with the same set of faces and the same external face. An \emph{(abstract) octilinear representation} contains all octilinear drawings with the same embedding that have the same angles between consecutive edges around each vertex and the same sequence of bends along an oriented version of each edge (e.g., from lower index to higher index endpoint).  
If a drawing fullfills these constraints, we say that it \emph{realizes} the representation. Since bends are predetermined, we can replace each bend by a vertex of degree $2$ and assume our representations have no bends. We now define formally \OctR and \textsc{Compaction}. 

\smallskip \noindent
\OctR: \emph{Given an octilinear representation $\mathcal{R}$, decide if $\mathcal{R} \neq \emptyset$.}

\smallskip \noindent
\OctC: \emph{Given a \emph{realizable} octilinear representation $\mathcal{R}$, report a drawing $\Gamma \in \mathcal{R}$ on the integer grid such that the area of $\Gamma$, i.e., the area of the smallest axis-parallel rectangle containing $\Gamma$, is at most the area of $\Gamma'$ for each $\Gamma' \in \mathcal{R}$ on the integer grid.}

\smallskip 
In the rest of the paper, unless otherwise stated, we denote by $G$ a graph and by $\mathcal{R}$ an octilinear representation of $G$. For the evaluation of the area in the \OctC problem, we have to make the use of the grid comparable. Observe that we are employing a standard model in graph drawing, that is, we require that vertices must be placed at integer coordinates. 
Moreover, also observe that for the input of \OctC, we require that $\mathcal{R} \neq \emptyset$ so to separate the computational complexity of the actual compaction from the NP-hardness of the necessary realizability. Note that even in this restricted setting, we show APX-hardness in Section~\ref{sec:inapprox}. 
We aim to apply the notion of limited non-convexity to investigate the parameterized complexity of \OctR and \OctC problems. Our results concern a set of more relaxed parameters that are related to the number of reflex corners. We consider four parameters: \begin{enumerate*}
    \item $\omega$ is the total number of reflex corners in the entire representation;
    \item $\varphi$ is the maximum number of faces  that contain at least one reflex corner;
    \item $\kappa$ is the maximum number of reflex corners per face; 
    \item $\delta$ is the number \emph{diagonals} (edges having slope $\pm 1$).
\end{enumerate*}

\section{Hardness Results}
\label{se:hard}

\subsection{Review of the NP-hardness Construction of Bekos et al.~\cite{DBLP:journals/algorithmica/BekosFK19}}
\label{sec:nphardnessliterature}

In both Sections \ref{sec:inapprox} and \ref{sec:paranp}, we extend the NP-hardness reduction from 3-SAT to \OctR and \textsc{Compaction} by Bekos et al.~\cite{DBLP:journals/algorithmica/BekosFK19} to obtain new results. Here, we briefly describe the approach of~\cite{DBLP:journals/algorithmica/BekosFK19}. As a central concept, information is encoded in the length $\ell_u$ of specific edges of the drawing. We first describe the proof for the easier case where $\ell_u=1$.
For reason of space, the approach to generalize to any value of $\ell_u$ 
will be only sketched in this section and described in details in Appendix \ref{app:hard}.

\begin{figure}[tb]
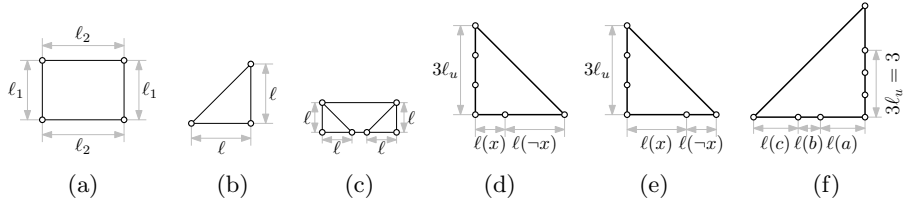

    \centering
    \begin{subfigure}[b]{0.18\textwidth}
        \centering
        \includegraphics[page=1,scale=0.7]{octilinear-np.pdf}
        \caption{}
        \label{fig:beyond-ortho-np:rectangle}
    \end{subfigure}
    \hfil
    \begin{subfigure}[b]{0.13\textwidth}
        \centering
        \includegraphics[page=3,scale=0.7]{octilinear-np.pdf}
        \caption{}
        \label{fig:beyond-ortho-np:triangle}
    \end{subfigure}
    \hfil
    \begin{subfigure}[b]{0.13\textwidth}
        \centering
        \includegraphics[page=4,scale=0.7]{octilinear-np.pdf}
        \caption{}
        \label{fig:beyond-ortho-np:copy}
    \end{subfigure}
    \hfil
     \begin{subfigure}[b]{0.16\textwidth}
        \centering
        \includegraphics[page=6,scale=0.7]{octilinear-np.pdf}
        \caption{}
        \label{fig:beyond-ortho-np:true}
    \end{subfigure}
    \hfil
    \begin{subfigure}[b]{0.16\textwidth}
        \centering
        \includegraphics[page=7,scale=0.7]{octilinear-np.pdf}
        \caption{}
        \label{fig:beyond-ortho-np:false}
    \end{subfigure}
    \begin{subfigure}[b]{0.2\textwidth}
        \centering
        \includegraphics[page=9,scale=0.7]{octilinear-np.pdf}
        \caption{}
        \label{fig:beyond-ortho-np:constantUnit}
    \end{subfigure}
      \caption{Gadgets used in the reduction from \textsc{$3$-SAT} to \textsc{Octilinear Realizability~\cite{DBLP:journals/algorithmica/BekosFK19}} for $\ell_u=1$: (a)~Propagation, (b)~rerouting, (c)~copy, (d)--(e)~variable and (f)~clause gadget. (d) and (e)~show $x=\top$ and $x=\bot$, respectivel.y
    }
    \label{fig:beyond-ortho-np:routing2}
\end{figure}

\smallskip
\noindent 
\textbf{Case $\mathbf{\ell_u=1}$: The gadgets.} There are three types of gadgets besides the clause and variable gadgets in~\cite{DBLP:journals/algorithmica/BekosFK19}. The \emph{propagation gadget} is a rectangular face, where the information can be propagated from either side to its opposite side. See Fig.~\ref{fig:beyond-ortho-np:rectangle}. The \emph{rerouting gadget} is a triangular face which allows to propagate from a vertical to a horizontal edge or vice-versa. See Fig.~\ref{fig:beyond-ortho-np:triangle}. The \emph{copy gadget} are two triangular faces connected with two straight-line edges. This gadget copies the length of an edge to three other edges letting information related to a variable reach all the clauses that contain the variable. See Fig.~\ref{fig:beyond-ortho-np:copy}. 

It remains to discuss how variables and clauses are encoded. The variable gadget for a variable $x$ consists of a single triangular shaped face as shown in Figs.~\ref{fig:beyond-ortho-np:true} and~\ref{fig:beyond-ortho-np:false}, so that the left side of the face is formed by three edges of length $\ell_u$ while the bottom side of the face consists of two edges. One of those is representing the literal $x$ while the other edge represents the literal $\neg x$. As a result, it holds that $\ell(x)+\ell(\neg x) = 3 \ell_u$.  It is easy to see that $\ell(x), \ell(\neg x) \in \{1,2\}$ on the integer grid. We assume that the literal whose corresponding edge has length $2$ to be true and check with the clause gadget in Fig.~\ref{fig:beyond-ortho-np:constantUnit}. For clause $(a \lor b \lor c)$ at least one among the three literals $a$, $b$, and $c$ is true, since the right side of the face has height more than three times the unit edge length $\ell_u=1$. See Fig.~\ref{fig:beyond-ortho-np:constantUnit}. See Fig.~\ref{fig:hard} for an example of reduction. 

\begin{figure}[tb]
    \centering
    \includegraphics[width=0.9\textwidth,page=5]{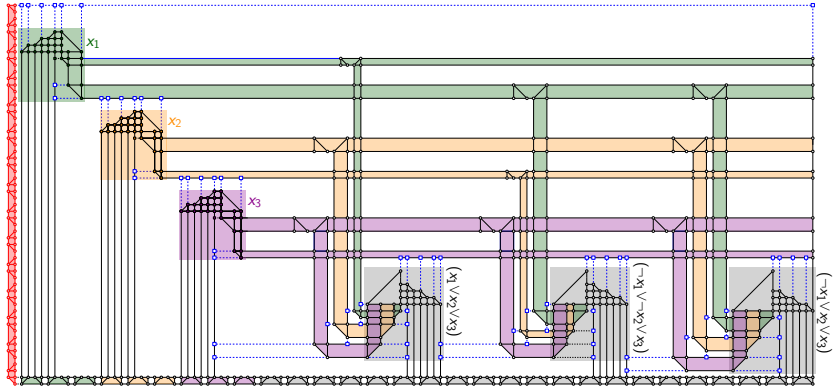}
    \caption{Example of the reduction descibed in~\cite{DBLP:journals/algorithmica/BekosFK19} with formula $(x_1\lor x_2 \lor x_3) \land (\neg x_1 \lor \neg x_2 \lor x_3) \land (\neg x_1 \lor x_2 \lor x_3)$ when $\ell_u= 1$ and the satisfying assignment $x_1=\bot$, $x_2 = x_3 = \top$ with our modifications (red and blue) as   described in the proof of Theorem~\ref{thm:approximation}. A clean version without our adjustments is in Appendix~\ref{app:hard}.}
    \label{fig:hard}
\end{figure}

\textbf{Case $\mathbf{\ell_u=1}$: Routing information.}  Refer to Fig.~\ref{fig:hard}.
The first $3n$ copy gadgets at the bottom side connect to the $n$ variable gadgets. The variable gadget of variable $x_i$ appears above the variable gadget of variable $x_{i+1}$. The output literals of each variable are rerouted using rerouting gadgets so that the variable length is propagated towards the right via a series of propagation and copy gadgets, to which we refer as a \emph{literal path} (shaded with three different colors in Fig.~\ref{fig:hard} and all the other figures). Note that if $\ell_u=1$, exactly one literal of each variable is true. For each clause there are seven copy gadgets at the bottom side of the drawing and a clause gadget; the last three copies of the unit length of the associated copy gadgets are propagated via propagation and rerouting gadgets to the clause gadget. The clause gadget receives the information from the corresponding literal paths via the right copied length of a copy gadget on the literal path. We bound the upper and right side of the drawing by inserting a vertex at the top right corner. We connect this vertex with the topmost copy gadget horizontally and a path through the ends of each variable path and the rightmost copy gadget vertically. 
See  Appendix~\ref{app:hard} for the general case $\ell_u\ge1$.


\subsection{Approximability of  Octilinear Compaction}
\label{sec:inapprox}

We investigate the parameterized complexity of \OctC  for $\omega = 4$ -- the outer face   contains at least $3$ reflex corners in any octilinear drawing.

\begin{restatable}{theorem}{inapprox}\label{thm:approximation}
There is no $\frac{9}{4}$-approximation for \OctC  even if $\omega = 4$ unless P=NP.
\end{restatable}

\begin{proof}[Sketch]
    We adapt the NP-hardness  construction by Bekos et al.~\cite{DBLP:journals/algorithmica/BekosFK19} for the special case $\ell_u = 1$. Namely, we add a sequence of copy gadgets to the left side of the construction (red shaded gadgets in Fig.~\ref{fig:hard}) which makes the area of the drawing dependent on $\ell_u$. In addition, we add some extra (subdivision) vertices and edges to make all faces convex (see blue parts  in Fig.~\ref{fig:hard}). See Appendix~\ref{app:inapprox} for details and computation of the approximation ratio.$\hfill\qed$
\end{proof}


\begin{corollary}\label{cor:noPTAS}
 Unless P=NP, there is no PTAS for \OctC.
\end{corollary}

\subsection{Para-NP-Hardness of \OctR}
\label{sec:paranp}

We shift our attention to the less restrictive parameters $\varphi$, the number of non-convex faces, and $\kappa$, the maximum number of reflex corners per face, and we prove para-NP-hardness even for  \OctR:

\begin{restatable}{theorem}{paraNP}
\label{thm:paraNP}
\renewcommand{\mylink}{\statlink{thm:paraNP}\prooflink{pparaNP}}
\OctR remains NP-hard for $\varphi=1$. Also, it remains NP-hard for $\kappa=8$.
\end{restatable} 

    
    
\begin{proof}[Sketch]
    We adapt the NP-hardness  proof by Bekos et al.~\cite{DBLP:journals/algorithmica/BekosFK19} for the case $\ell_u\ge 1$. For $\varphi=1$, we achieve that only the outer face contains reflex corners. Then, for $\kappa=8$, we subdivide the outer face so that each gadget requiring reflex corners is contained in its own face. See Appendix~\ref{app:paranp} for details. 
    $\hfill\qed$
\end{proof}


\section{Parametrized Results}
\label{se:param}
In this section we show that \OctR and \OctC can be efficiently solved when there are either few reflex angles and diagonal edges, respectively. Namely, in  Sections \ref{sec:fpt} we prove that there exists an FPT algorithm on parameter $\omega$ for  \OctR and in Section  \ref{sec:xp} we describe an XP algorithm on parameter~$\delta$ for  \OctC problem. Before that, in Section \ref{se:flow}, we describe a flow procedure that will be used in the other two sections to prove the respective results.

\subsection{Flow algorithm for the convex case.}

Assume that in $\mathcal{R}$ all inner faces are convex and that the outer face contains no inflex angle (that is, the polygon bounding the outer face is convex as well). In this case, a realizing drawing exists if and only if for each face the sum of widths of the edges forming the top boundary is equal to the sum of the widths of the edges forming the bottom boundary and the same holds for the left and the right boundary.  
The reason for that is that one can use the obtained edge lengths to construct the drawing after fixing an arbitrary vertex to an arbitrary point in the plane. In any such solution, for each face, the edges forming the facial cycle are indeed realized as a crossing-free convex polygon. Thus, the resulting drawing correctly realizes the given planar embedding and it is therefore planar.

We can model this behavior using two auxiliary flow networks similar to techniques used in orthogonal graph drawing~\cite{DBLP:books/ph/BattistaETT99,DBLP:journals/siamcomp/Tamassia87}. Namely, $G^\rightarrow$ (blue graph in Fig.~\ref{fig:convex_representation}) contains a node for each internal face and two nodes $s$ and $t$ for the outer face. In addition, arc $(f_1,f_2)$ exists if there is an edge $e$ in $\mathcal{R}$ so that $f_1$ occurs to the left of $e$ and $f_2$ to the right (for the outer face all outgoing and incoming edges are attached to $s$ and $t$, resp.). It is easy to see that any valid flow from $s$ to $t$ with positive value along each arc describes an assignment of horizontal lengths to the edges of $\mathcal{R}$ satisfying the height constraint stated above. Similarly, we can define $G^\downarrow$ for the width constraint (red graph in Fig.~\ref{fig:convex_representation}). These two network flows can be described as a linear program. 
   
   To be more precise, we create a linear program formulation of the flow as follows. We associate each arc $a=(f_1,f_2) \in E(G^\rightarrow) \cup E(G^\downarrow)$  with a variable $x_a$ that must have strictly positive values, i.e., $x_a \geq 1$. Then, for each face $f$ that is not the outer face we create two constraints that ensure flow conservation in both $G^\rightarrow$ and $G^\downarrow$, namely,  $\sum_{a=(f,f') \in E(G^\rightarrow)}x_a=\sum_{a=(f',f) \in E(G^\rightarrow)}x_a$ and $\sum_{a=(f,f') \in E(G^\downarrow)}x_a=\sum_{a=(f',f) \in E(G^\downarrow)}x_a$.

   Diagonal segments create an arc in both $G^\rightarrow$ and $G^\downarrow$ (orange highlights in Fig.~\ref{fig:convex_representation}). Thus, we must additionally require the flows along these two arcs to be the same as the height and the width of a diagonal at slope $\pm 1$ are the same. These additional constraints are easily included in our linear program; e.g., we can simply use the same variable for both arcs.

   The resulting linear program can be solved using polynomial time algorithms which yields a rational solution that can be encoded with a polynomial number of bits~\cite{DBLP:journals/combinatorica/Karmarkar84} that can then be scaled the to an integer solution in polynomial time.

    We first state the following lemma, which is an interesting result by itself and that we use as the main tool in our FPT algorithm and XP algorithms. The lemma is directly implied by the described linear program modelling the flow.

\label{se:flow}

   \begin{figure}[t]
    \centering
    \includegraphics[scale=.7,page=1]{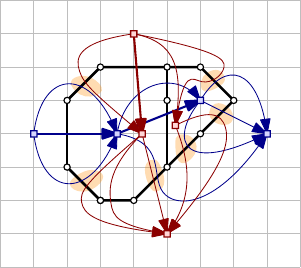}
    \caption{A representation $\mathcal{R}$ and auxiliary flow networks $G^\rightarrow$ (blue) and $G^\downarrow$ (red). The value of flow along all thicker arcs is 2, the value of flow along all other arcs is 1. The positions on the grid corresponds to the solution yielded.}
    \label{fig:convex_representation}
\end{figure}

\begin{lemma}\label{thm:convex}
 \OctR is in P when there is no reflex angle in any internal face and no inflex angle in the outer face in $\mathcal{R}$.
\end{lemma}

\subsection{An FPT Algorithm for Octilinear Realizability}
\label{sec:fpt}

Given Lemma~\ref{thm:convex}, our goal is now to extend it for fixed values of $\omega$. In order to do that, we need further notation and intermediate results. Recall that $\mathcal{R}$ denotes an octilinear representation of a graph $G$.  We describe how to enrich  $G$ and $\mathcal{R}$ so that each face containing at least one reflex angle has size $O(\omega)$. These enrichments are denoted by the  \emph{shadow graph} and \emph{representation}. The key property of a shadow representation $\mathcal{R}'$ is that such a representation is realizable if and only if $\mathcal{R}$ is realizable (Lemma~\ref{le:cages}). Thus, we can work~on the shadow graph and its shadow representation in order to obtain an FPT algorithm.

\paragraph{Shadow Graph and Representation}
\label{subse:shadow}
We begin by defining the \emph{shadow graph} $G'$ and an octilinear representation $\mathcal{R}'$ of $G'$, called the \emph{shadow representation}, which are obtained from the input graph $G$ and its input octilinear representation $\mathcal{R}$, respectively. Initially, we set $G'=G$ and $\mathcal{R}'=\mathcal{R}$. During the construction of the shadow graph and representation, we analyze each face $f$ and enrich the graph considering the vertices incident to $f$ and their angles inside $f$. 

\begin{figure}[tb]
    \centering
    \begin{subfigure}[b]{0.45\textwidth}
        \centering
        \includegraphics[page=1,width=.8\textwidth]{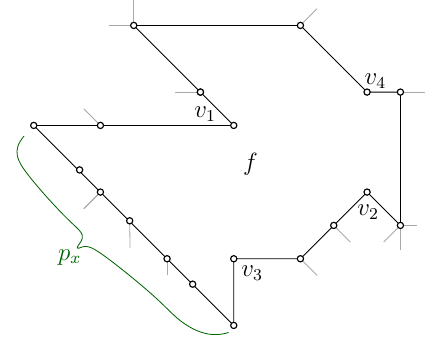}
        \caption{$G$ and $\mathcal{R}$}
        \label{fig:cage-1}
    \end{subfigure}
    \hfil
    \begin{subfigure}[b]{0.45\textwidth}
        \centering
        \includegraphics[page=2,width=.8\textwidth]{shadow_ext.pdf}
        \caption{$G'$ and $\mathcal{R}'$}
        \label{fig:cage-2}
    \end{subfigure}
    \caption{The construction of $G'$ and $\mathcal{R'}$.}
    \label{fig:cage}
\end{figure}

Consider first the internal faces. For each internal face $f$ of $G$, consider each vertex $v$ incident to $f$ and add a vertex $v'$ if the angle of $v$ in $f$ is either strictly convex or reflex. We say that $v'$ is the \emph{shadow vertex} of $v$. We order the shadow vertices following the order of the corresponding vertices and add a cycle connecting them. Let $C_f$ denote this cycle. See the thicker red cycle in Fig.~\ref{fig:cage-2}.

The newly created face $f'$ bounded by $C_f$  is called the \emph{shadow face} of $f$. The angle inside $f'$ of shadow vertex $v'\in C_f$ in $G'$ is the same as the angle inside $f$ of the vertex $v$ in $G$. Finally, we use horizontal and vertical edges and, if necessary, subdivision vertices in order to connect $C_f$ to the border of $f$ in $G'$. We do this as follows. Consider vertex $v$, its shadow vertex $v'$, and the angle $\alpha$ at $v$ inside $f$ in $\mathcal{R}$, which is the same as the angle of $v'$ in $\mathcal{R}'$. Consider the angle $\alpha_f$ at $v'$ in the face between the boundary of $f$ and $C_f$. Note that either $\alpha$ is reflex and $\alpha_f$ is inflex or vice versa. We describe the construction for the case where $\alpha$ is reflex, the other case can be handled similarly. We proceed with a case analysis, refer to Fig.~\ref{fig:cage} where Fig.~\ref{fig:cage-1} depicts $f$ in $G$ also showing the configurations of the angles described in $\mathcal{R}$, while Fig.~\ref{fig:cage-2} shows the following construction.  

If $\mathbf{\alpha=\frac{7}{4}\pi}$,  $v$ is incident to either a horizontal or a vertical edge on the border of $f$. Consider the first case, the second can be treated similarly. We consider the diagonal edge incident to $v'$ and we subdivide it with a vertex $s$. Then we add a horizontal edge from $v$ to this subdivision vertex, fixing the angles accordingly. See $v_1$ in Fig.~\ref{fig:cage-2}. If  $\mathbf{\alpha=\frac{3}{2}\pi}$, $v$ is either incident to two diagonal or to a horizontal and a vertical edge. If $v$ is incident to two diagonal edges, one of its horizontal ports or one of  its vertical ports is free inside $f$: We add an edge connecting $v$ to $v'$ vertically (see $v_2$ in Fig.~\ref{fig:cage-2}) or horizontally, respectively. If $v$ is incident to a horizontal and a vertical edge, we subdivide the horizontal edge of $C_f$ incident to $v'$ with a subdivision vertex $s$ and we add a vertical edge from $v$ to $s$. See $v_3$ in Fig.~\ref{fig:cage-1}. If $\mathbf{\alpha=\frac{5}{4}\pi}$,  e add an edge connecting $v$ to $v'$ that is vertical if $v$ is incident to a horizontal edge or horizontal otherwise. See $v_4$ in Fig.~\ref{fig:cage-2}. If $\mathbf{\alpha < \pi}$, then $\alpha$ is strictly convex and $\alpha'$ is reflex, i.e., we do a symmetric operation. More precisely, we subdivide the facial cycle of $f$ if necessary and select the orientation of the edge connecting $v$ and $v'$ based on the edges incident to the $v'$.





Consider now the outer face $f_\circ$. We construct $C_{f_\circ}$ similarly to how we constructed $C_f$ for an internal face, inverting the construction as in this case $C_{f_\circ}$ is external to $f_\circ$ while $C_f$ was internal with respect to $f$. Refer to Fig.~\ref{fig:extension-2}, that depicts an octilinear drawing realizing the shadow representation $\mathcal{R}'$ of the shadow graph $G'$, where $G$ and $\mathcal{R}$ are depicted in Fig.~\ref{fig:extension-1}. In the following, we denote by $G'$ and $\mathcal{R}'$ the shadow graph and representation of $G$ and $\mathcal{R}$.

\begin{figure}[tb]
    \centering
    \begin{subfigure}[t]{0.3\textwidth}
        \centering
        \includegraphics[page=3,scale=.4]{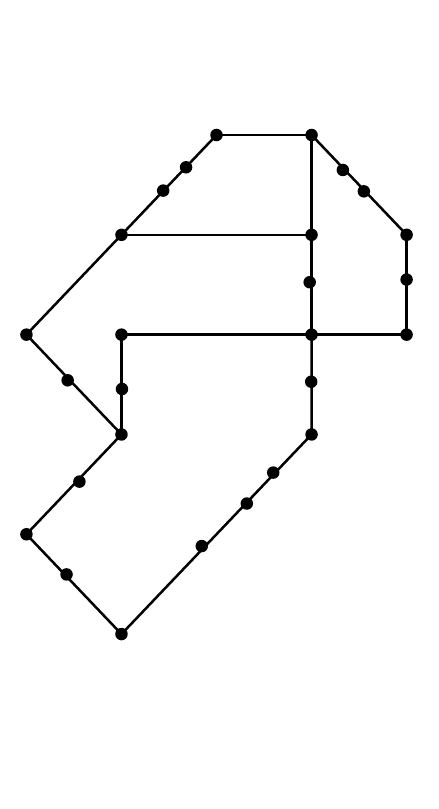}
        \caption{$G$ and $\mathcal{R}$}
        \label{fig:extension-1}
    \end{subfigure}
    \hfil
    \begin{subfigure}[t]{0.3\textwidth}
        \centering
        \includegraphics[page=4,scale=.4]{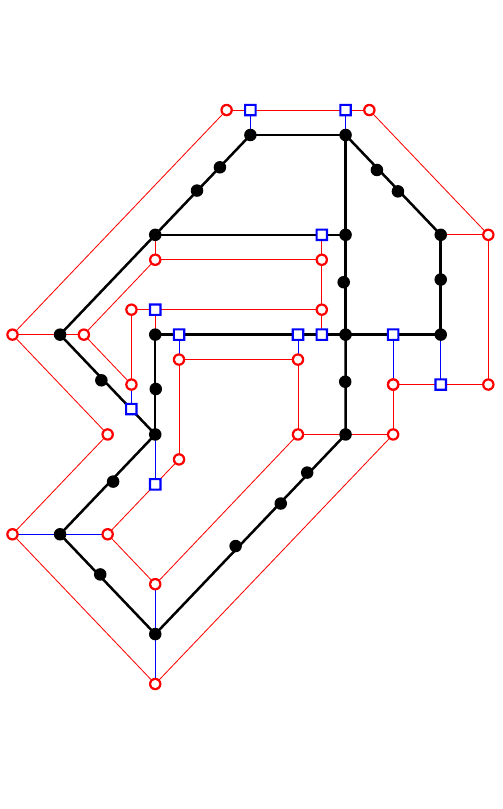}
        \caption{$G'$ and $\mathcal{R}'$}
        \label{fig:extension-2}
    \end{subfigure}
    \hfil
    \begin{subfigure}[t]{0.38\textwidth}
        \centering
        \includegraphics[page=5,scale=.4]{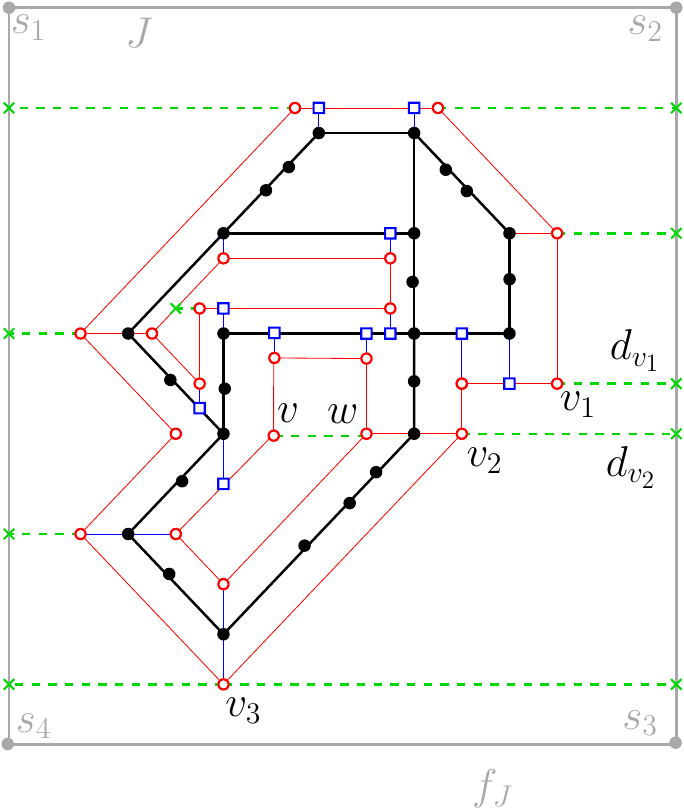}
        \caption{$G''$ and $\mathcal{R}''$}
        \label{fig:extension-3}
    \end{subfigure}
    \caption{(a)~Graph $G$ and its representation $\mathcal{R}$. 
    (b)~Shadow graph $G'$ and shadow representation $\mathcal{R}'$ obtained from Graph $G$ and its representation $\mathcal{R}$. (c)~Planar extension $\mathcal{R}''$ of $\mathcal{R}'$.}
    \label{fig:extension}
\end{figure}


\begin{lemma}
\label{le:cages}
$\mathcal{R'}\not = \emptyset$ if and only if $\mathcal{R}\not = \emptyset$.
\end{lemma}

\begin{proof}
Given an octilinear  drawing of $G'$ one can obtain an octilinear drawing of $G$ by removing the shadow vertices and their incident edges; given an octilinear drawing of $G$ one can obtain a drawing of $G'$ by adding the vertices of $C_f$ and the corresponding edges following the construction of $\mathcal{R}'$, eventually scaling the drawing to create space for such new vertices and edges when considering the internal faces. To this end, also note that the lengths of edges connecting vertices of $f$ with vertices of $C_f$ can be chosen arbitrarily small.$\hfill\qed$
\end{proof}

For every face of the shadow graph containing a reflex angle the number of vertices forming $\pi$ angles inside is bounded by the number of reflex angles inside (as the angles forming $\pi$ angles are exactly the angles created by subdivisions when defining the face of the shadow graph). Thus, we have the following:

\begin{lemma}
\label{le:facesbounded}
The total length of the non-convex faces of $G'$ is $O(\omega)$.
\end{lemma}
\begin{proof}
Observe that an octilinear representation with $O(\omega)$ reflex angles has $O(\omega)$ inflex angles. Concerning a face $f$ of $G$ with $\omega_f$ reflex angles, we could still have $\Omega(n)$ vertices on its boundary, since many of them could form  angles of $\pi$ inside $f$. The same does not hold by construction for the shadow faces in $G'$. By construction, a shadow face $f'$ of $G'$ constructed for a face $f$ with $\omega_f$ reflex corners has $O(\omega_f)$ vertices as only strictly convex, i.e., inflex, and reflex corners plus subdivision vertices that were introduced together with reflex corners (at most one subdivision vertex per reflex corner)  define its boundary. All the vertices forming an angle of $\pi$ inside $f$ correspond to an edge in $f'$; see for instance the path $p_x$ that is incident to $f$ in Fig.~\ref{fig:cage-1} and that corresponds to the edge $e_x$ incident to $f'$ in Fig.~\ref{fig:cage-2}. By construction, all the faces of $G'$ incident to vertices of $G$ that were forming angles of $\pi$ in the original representation inside $f$ have no reflex angle inside. Hence, the total length of the shadow faces, summing over all faces $F_r$ with reflex corners is $\sum_{f\in F_r}O(\omega_f)=O(\omega)$. $\hfill\qed$
\end{proof}

\paragraph{Planar Extensions and FPT Algorithm}
\label{subse:fpt-alg}
Given the shadow graph $G'$ and its shadow representation $\mathcal{R}'$ as introduced in Section~\ref{subse:shadow}, an \emph{extension} of  $\mathcal{R}'$ (and consequently  of $G'$, which is the graph that $\mathcal{R}'$ represents) is defined as follows. We first add a 4-cycle $J$ consisting of four additional vertices $s_1$, $s_2$, $s_3$, and $s_4$ in $G'$. The angles around these vertices are $\frac{\pi}{2}$ inside the cycle and $\frac{3\pi}{2}$ outside. 
We define the outer face $f_J$ of $G'$ to be equal to the face of this cycle having the reflex angles outside. Hence, the rest of the graph has to be inside  $J$. Let $G''$ be the plane graph obtained so far, which consists of $G'$ and $J$. Let $\mathcal{R''}$ be the representation of $G''$ obtained as described above.

Consider $G''$ and $\mathcal{R}''$ and a shadow face $f'$ of $G''$, whose border is $C_f$. Let $v$ be a vertex of $C_f$ forming a reflex angle in $f'$.  Let $w$ be another vertex of $C_f$. Let $e$ be an edge of $C_f$ or, if $f$ is the outer face of $G$, it could be one of the four edges of $J$. Observe that in this case $v$ is incident to the face between the two disconnected components $J$ and $G'$ of $G''$.
We define an operation that we call \emph{alignment}, that takes as input either $v$ and $w$ or $v$ and $e$: In the first case, we add an edge connecting $v$ to $w$; in the second case, we subdivide $e$ and we add an edge connecting $v$ to the subdivision vertex $d_v$ added in $e$. In both cases, we change the angles in the representation accordingly so that the new edge is horizontal. 
Fig.~\ref{fig:extension-3} shows the octilinear representation of $\mathcal{R'}$ depicted in Fig.~\ref{fig:extension-2} after performing alignments (added edges are dashed). Vertex $v$ is aligned to $w$ and both $v_1$ is aligned to $(s_2,s_3)$, the respective dummy vertices are $d_{v_1}$ and $d_{v_2}$.

We now define the pair $\langle G_{+},\mathcal{R}_+\rangle$, constructed starting from $G''$ and $\mathcal{R}''$ and iteratively choosing one possible alignment for a vertex $v$ forming a reflex angle in a face of $\mathcal{R}''$. When two or more vertices align at the same edge $(a,b)$, we chose an ordering of the subdivision vertices associated to such vertices when traversing such subdivided edge from $a$ to $b$. 
A single vertex can be considered at most twice by this procedure, as adding two horizontal edges to $\mathcal{R''}$ in any vertex makes sure that no angle around the vertex is larger than $\pi$. See for example vertex $v_3$ in Fig.~\ref{fig:extension-3}. Also, the considered ordering could imply that the obtained graph with the given rotation system at the vertices and the given outer face is not planar. For example, consider $v_1$ and $v_2$ that are both aligned to edge $(s_2,s_3)$, in Fig.~\ref{fig:extension-3}. If the order of $d_{v_1}$ and $d_{v_2}$ was switched, 
the corresponding graph with the given rotation system and $f_J$ as outer face would not be planar.


Given one of the possible extensions $\langle G_{+},\mathcal{R}_+\rangle $, of $G'$ and $\mathcal{R'}$ we have a graph that might be not planar given that an outer face and that the rotation system of the vertices in $\mathcal{R'}$ is fixed. We say that an extension is \emph{planar} if the graph $G'_+$ with the rotation system defined by $\mathcal{R}'_+$ is planar and $J$ as the cycle bounding the outer face. We prove the following two observations in Appendix~\ref{app:fpt}. 

\begin{restatable}{lemma}{fewext}
\label{le:fewext}
\renewcommand{\mylink}{\statlink{le:fewext}\prooflink{le:fewextapp}}
There exists $2^{O(\omega^2)}$ possible extensions of $G'$ and $\mathcal{R'}$.
\end{restatable} 
%
%
%
\begin{restatable}{lemma}{extrealiz}
\label{le:extrealiz}
$\mathcal{R}'\not= \emptyset$ $\Leftrightarrow$ There exists an extension $\mathcal{R}_+'\not = \emptyset$ of $\mathcal{R}'$.
\end{restatable}

We are now ready to combine our results into an FPT algorithm.

\begin{restatable}{theorem}{ftprealizability}
\label{th:ftprealizability}
\renewcommand{\mylink}{\statlink{th:ftprealizability}\prooflink{pftprealizability}}
\OctR is FPT with respect to $\omega$.
\end{restatable}


\begin{proof}
Let $\mathcal{R}$ denote an octilinear representation of a graph $G$. If $G$ has a vertex with degree at least $9$ or  if it is not planar, we reject the instance. Otherwise, we construct the shadow graph $G'$ of $G$ and the shadow representation $\mathcal{R'}$ of $\mathcal{R}$. Then we consider all the possible extensions of $\mathcal{R}'$ and we test the realizability of each extension using Lemma \ref{thm:convex}, which can be done as each extension has only four angles in the outer face, which are reflex angles, and no reflex angle inside by construction. We have that one of such extensions is realizable if and only if $\mathcal{R}$ is realizable by   Lemmas \ref{le:cages} and~\ref{le:extrealiz}. We now discuss the running time of this procedure. In order to construct $G'$ and $\mathcal{R}'$ it suffices to traverse each face of $G$ $O(1)$ times and, for each vertex and each edge, to perform $O(1)$ operations consisting in the case analysis depending on the angle at that vertex inside the face. Such operations are vertex addition, edge addition, and edge subdivision operations. Hence, constructing $G'$ and $\mathcal{R'}$ can be done in polynomial time. To this end, note in particular, that the total number of vertices in the shadow graph is upper bounded by $17n$ as each vertex is incident to at most $8$ faces and for each face we may create a shadow vertex plus a subdivision vertex in $G'$. This also limits the total number of edges of the shadow graph to $51n-6$. Both of these values are independent of $\omega$. Next, there are $2^{O(\omega^2)}$ possibilities to for the extension of $\mathcal{R'}$. On the other hand, each of these extensions can be done in polynomial time and the consecutive realizability test using $\mathcal{R'}$ runs in polynomial time. Thus,  the computational time of the entire procedure is of type $f(n) 2^{O(\omega^2)}$ where $f(n)$ is a polynomial function on variable $n$.   $\hfill\qed$
\end{proof}


\subsection{An XP Algorithm for Octilinear Compaction}
\label{sec:xp}

We assume  there is no reflex angle in any internal face and no inflex angle in the outer face in $\mathcal{R}$. Observe that $\omega\in \{3,4\}$ in this case since the outer face always contains at least three reflex corners. Moreover, by definition of \OctC, we may assume that $\mathcal{R}$ admits a realizing drawing. Note that this property can also be tested using Theorem~\ref{th:ftprealizability} since $\omega\in \{3,4\}$. We now describe an XP algorithm on parameter $\delta$.  Our goal is first proving that there exists a drawing of $G$ realizing $\mathcal{R}$ such that the area of $\Gamma$ depends exponentially only on the variable $\delta$. In order to do that, given $G$ and $\mathcal{R}$, we define the \emph{skeleton graph} $G^*$ and the \emph{skeleton representation}  $\mathcal{R}^*$. Refer to Fig. \ref{fig:octy-xp}. 

We define $\mathcal{R}^*$ given $\mathcal{R}$, the definition of $G^*$ given $G$ is then implicit. The vertices in $\mathcal{R}^*$ are the ones incident to at least one diagonal edge in $\mathcal{R}$. These vertices are represented by squares in Fig.\ref{fig:octy-xp} and diagonals are represented thicker. Let $u$ and $v$ be two vertices of $\mathcal{R}^*$: We add an edge with slope $+1$ or $-1$ if such edge is also present in $\mathcal{R}$; we add a horizontal (resp. vertical) edge if there is an horizontal (resp. vertical) path connecting $u$ and $v$ in $\mathcal{R}$. 
\begin{figure}[tb]
  \begin{subfigure}[t]{0.48\textwidth}
        \centering
        \includegraphics[page=1,scale=.65]{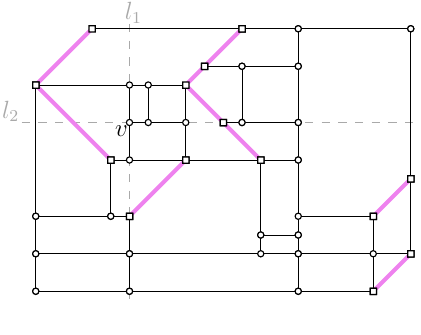}
        \caption{$G$ and $G^*$}
        \label{fig:octy-xp-a}
    \end{subfigure}
    \hfil
    \begin{subfigure}[t]{0.48\textwidth}
        \centering
        \includegraphics[page=2,scale=.65]{octy-xp.pdf}
        \caption{$\mathcal{R}$ and $\mathcal{R}^*$}
        \label{fig:octy-xp-b}
    \end{subfigure}
    \caption{(a) The graph $G$ and its representation $\mathcal{R}$. (b) The compressed graph and representation $G^*$ and $\mathcal{R}^*$. The figures illustrate the proof of Lemma~\ref{le:compressed-drawing}.}
    \label{fig:octy-xp}
\end{figure}

\begin{lemma}
\label{le:compressed-drawing}
There exists a drawing of $G$ realizing $\mathcal{R}$ with area $O(n^2 \cdot A(\delta))$, where $A(\delta)$ is a computable function.
\end{lemma}

\begin{proof}
Let $\Gamma$ be a drawing realizing $\mathcal{R}$. We remove edges and vertices and smooth vertices from $\Gamma$ to obtain a drawing of a graph isomorphic to $G^*$. By construction of $\mathcal{R}^*$, this is a drawing of $G^*$ realizing $\mathcal{R}^*$. We denote by $\Gamma^*$ such drawing.  Since $G^*$ has $O(\delta)$ vertices, we can also assume that $\Gamma^*$ has an area that only depends on $\delta$. Let $h^*(\delta)$ and $w^*(\delta)$ the height and width of the smallest axis-parallel rectangle containing $\Gamma^*$. We scale $\Gamma^*$ so that the minimum height or width of any of its edges is at least $n$. The smallest axis-parallel rectangle containing the drawing $\Gamma^*$, computed as described above, has now width $O(w^*(\delta) \cdot n)$ and height $O(h^*(\delta) \cdot n)$.  We can now construct a realization $\tilde{\Gamma}$ of $\mathcal{R}$ different from $\Gamma$, using $\Gamma$ to enrich $\Gamma^*$ and then we prove that its area is $O(n^2 \cdot A(\delta))$.

We first consider every horizontal line of $\Gamma$ containing at least a vertex and we map it to an horizontal line of $\Gamma^*$ so that: The set of crossings of the diagonal edges the two lines have in the two drawings is the same (observe that every diagonal edge in $\Gamma$ corresponds exactly to a diagonal edge of $\Gamma^*$); the order of the horizontal lines is the same in the two drawings. We do the same for the vertical lines. We now consider $\Gamma^*$ and the horizontal and vertical lines described so far, and we place every vertex of $G$ in the intersection of the corresponding vertical and horizontal line (either subdividing an edge or placing it incident to no edges). See, for example, vertex $v$ in Fig. \ref{fig:octy-xp}: The lines $l_1$ and $l_2$ are mapped to the lines $l_1^*$ and $l_2^*$, respectively, and $v$ is placed according to the vertical line $l_1^*$ and the horizontal line $l_2^*$. This construction is always possible, as every edge of $\Gamma^*$ has height or width (or both, in case of diagonal edges) equal or higher than $n$. Finally, after all the vertices are placed, we add the remaining edges. Let $\tilde{\Gamma}$ be the obtained drawing, which by construction is a realization of $\mathcal{R}$. The area of $\tilde{\Gamma}$ is  $O(w^*(\delta) \cdot n) \times O(w^*(\delta) \cdot n)=O(n^2 \cdot A(\delta))$ with $A(\delta)=h^*(\delta \cdot w^*(\delta))$.$\hfill\qed$
\end{proof}

Lemma \ref{le:compressed-drawing} implies in particular that we can assume, without loss of generality, that there exists a realization of $\mathcal{R}$ where every diagonal has height (and width) at most $A(\delta)$. We can consequently consider all the possible combinations of height (and width) of the diagonals and, for each combination use the flow algorithm described in Section \ref{se:flow} to minimize the area, once the flow along the diagonal is fixed. Among all the drawings computed with this procedure, we finally select the one having minimum area. As this procedure is exhaustive, this is the drawing realizing $\mathcal{R}$ having minimum area. Hence, we have the following:

\begin{theorem}
\label{th:xp}
\OctC is XP when there is no reflex angle in any internal face and no inflex angle in the outer face in $\mathcal{R}$. 
\end{theorem}

\begin{remark}
    Lemma~\ref{thm:convex} implies that $A(\delta)$ can be encoded in $\mathcal{O}(\delta)$ bits, i.e., it is at most exponential in $\delta$. Also observe that the conditions of Theorem~\ref{th:xp} imply $\omega \in \{3,4\}$, that is, our XP-algorithm solves the problem in general when $\omega=3$.
\end{remark}

\section{Open Problems}
\label{sec:Open}


For some $\kappa < 8$, it may be that \OctR becomes polynomial time solvable. In addition, our FPT algorithm may be parameterizable for a suitable definition of the kitty corners concept that is used for orthogonal graph drawings.
Finally, it remains open whether there is a constant factor approximation for \OctC or if the problem is inapproximable.

\bibliography{references}

\newpage
\appendix

\newpage 

\section{Omitted Details from Section~\ref{sec:nphardnessliterature}}
\label{app:hard}

\begin{figure}[h!]
    \centering      \includegraphics[width=0.9\textwidth,page=6]{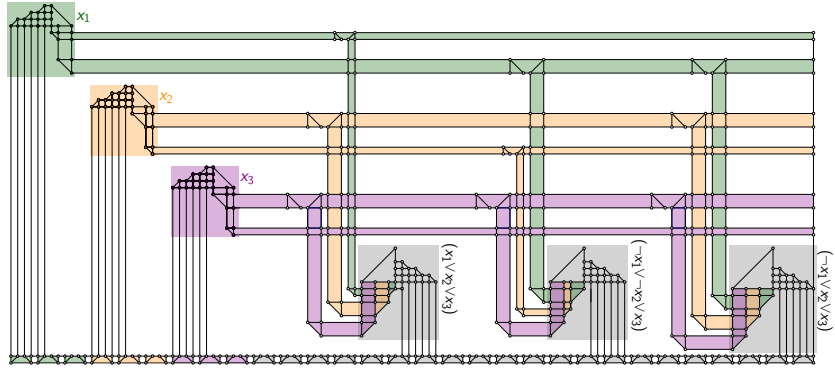}
    \caption{The construction from Fig.~\ref{fig:hard} without our adjustments from Section~\ref{sec:inapprox}.}
    \label{fig:hard-clean}
\end{figure}

\begin{figure}[b]
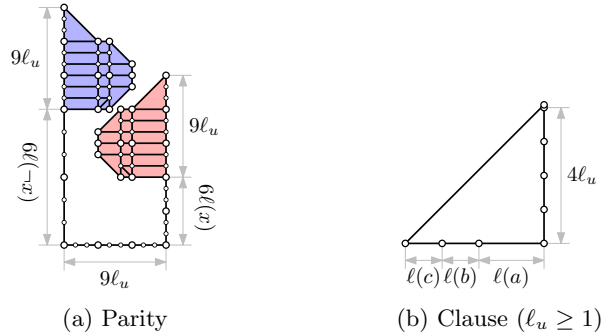

    \centering
    \begin{subfigure}[b]{0.24\textwidth}
        \centering
        \includegraphics[page=13,scale=0.8]{octilinear-np.pdf}
        \caption{Parity}
        \label{fig:beyond-ortho-np:parityOctilinear}
    \end{subfigure}
    \hfil
    \begin{subfigure}[b]{0.24\textwidth}
        \centering
        \includegraphics[page=10,scale=0.8]{octilinear-np.pdf}
        \caption{Clause ($\ell_u\ge1$)}
        \label{fig:beyond-ortho-np:variableUnit}
    \end{subfigure}
      \caption{Additional gadgets used in the reduction from \textsc{$3$-SAT} to \textsc{Octilinear Realizability~\cite{DBLP:journals/algorithmica/BekosFK19}} for the general case $\ell_u\ge1$.
    }
    \label{fig:beyond-ortho-np:routing}
\end{figure}

\textbf{General case $\mathbf{\ell_u\ge1}$.} Observe that the discussion above does not immediately generalize to arbitrary values of $\ell_u$. For instance, if $\ell_u=2$, we can set $\ell(x)=\ell(\neg x)=3$ for each variable $x$. Then, in the clause gadget for $(a\lor b\lor c)$, we have $\ell(a)+\ell(b)+\ell(c)=9>6=3\ell(u)$. Hence, in~\cite{DBLP:journals/algorithmica/BekosFK19} the authors introduce a fourth type of gadget.  The construction of the \emph{parity gadget} is more involved see Fig.~\ref{fig:beyond-ortho-np:parityOctilinear}. For each variable $x$ there is a parity gadget. The boundary of a parity gadget is defined by edges of lengths $\ell_u$, $\ell(x)$ and $\ell(\neg x)$ and produces a crossing between the two \emph{blocks} (highlighted blue and red in Fig.~\ref{fig:beyond-ortho-np:parityOctilinear}) unless $\ell(x),\ell(\neg x) \in \ ]0,1.084\ell_u[ \ \cup\  ]1.916 \ell_u,3\ell_u[$ (where $]a,b[$ denotes the open interval between $a$ and $b$); for a proof, see~\cite[Theorem 5]{DBLP:journals/algorithmica/BekosFK19}. The parity gadget makes sure the values of $\ell(x)$ and $\ell(\neg x)$ are in the precise range described above. Since edges corresponding to false literals can be slightly longer than $\ell_u$, the sum of the edge lengths of the three literals must be more than $4\ell_u$ which can even be achieved if only one literal is true.  See Fig.\ref{fig:beyond-ortho-np:variableUnit}. See Fig. \ref{fig:parameterized-np-hardness-large-a} for an example of the reduction described in ~\cite{DBLP:journals/algorithmica/BekosFK19}. See Fig.\ref{fig:octilinear-np-app} for an enlarged version of the figure, sliced into the separate gadgets. 

\smallskip \noindent
\textbf{Details about the construction for $\ell_u\ge 1$} Refer to Fig.~\ref{fig:parameterized-np-hardness-large-a} and, for more enlarged figures, also to Appendix \ref{app:paranp}. There is a chain of copy gadgets on the bottom boundary of the drawing, so that each of the gadgets has two copies of the copied edge length at its top side. We will interpret the copied edge length as the unit length of the drawing. The length of this chain will be implicitly defined in the following discussion. Then, we put the variable gadgets at the left side of the construction. Namely, we place the variable gadget for variable $x_i$ above and to the left of the corresponding gadget for variable $x_{i+1}$. We connect each variable gadget to two copy gadgets on the bottom side of the drawing, that is, one copy remains unused.  We then reroute the two literals of variable $x_i$ so that a sequence of propagation and copy gadgets can propagate it rightwards to its usages within the parity and clause gadgets. Note that these literal paths of variable $x_i$ occur above the variable gadget of variable $x_{i+1}$.  Next, we place the parity gadgets of all variables on the top side of the drawing, with the two blocks of the gadget being part of the outer face, above the variable gadget of $x_1$ so that the parity gadget of variable $x_i$ appears to the left of the parity gadget of variable $x_{i+1}$. We connect each parity gadget to the two corresponding literal paths via three copy gadget on each of the paths. In addition, the unit length is obtained from fourteen copy gadgets that copy the unit length on the bottom side of the drawing; again one copy each remains unused. Finally, it remains to position the clause gadgets to the right of the parity gadget of variable $x_n$. Again, we place the gadget for clause $c_i$ to the left of the gadget for clause $c_{i+1}$ so that the diagonal segment is on the outer face. Moreover, we connect the clause gadget of clause $c_i$ to the three contained literal paths via a copy gadget on each of the paths (note that a copy remains unused). The unit length is provided from two copy gadgets on the bottom side of the drawing where one copy remains unused.


\section{Omitted Details from Section~\ref{sec:inapprox}}
\label{app:inapprox}

\inapprox*

\begin{proof}
We reduce from 3-SAT and closely follow the construction by Bekos et al.~\cite{DBLP:journals/algorithmica/BekosFK19} with the variant where we assume that the unit length $\ell_u$ is equal to $1$. This property is maintained by making the size of the output drawing dependent on the unit length $\ell_u$. Namely, at the left boundary of the outer face, there exist $3n+5$ copy gadgets while at the bottom boundary of the outer face, we have $3n+7m$ copy gadgets where $n$ and $m$ are the number of variables and clauses of the $3$-SAT formula, respectively; see Fig.~\ref{fig:hard} for an example, where the copy gadgets added w.r.t. the reduction of~\cite{DBLP:journals/algorithmica/BekosFK19} (Fig.~\ref{fig:hard}) are on the left and shaded red. 

There are some reflex corners in the drawing. However, we can add another horizontal or vertical segment by shooting a ray towards the face for which the vertex is reflex and subdividing the first segment crossed by the ray; see blue dashed segments and blue square-shaped vertices in Fig.~\ref{fig:hard}. The obtained representation in fact has only convex internal faces and $\omega=4$ as the unbounded outer face has four reflex corners. Since we inserted all the required gadgets for the NP-hardness reduction, we can conclude that there is no efficient algorithm computing a drawing with $\ell_u=1$ for the obtained representation unless P=NP. 
It remains to discuss the area. The number of copies is chosen, so that all free edge lengths of the copy gadgets and the lengths of the edges connecting two copies of those gadgets at the top and bottom side can be chosen to be $1$. Thus, each copy gadget is $2\ell_u+1$ wide, in addition, between two copies there is an edge of length $1$. Thus, depending on $\ell_u$, the width and height of the drawing are $h(\ell_u) = (2\cdot \ell_u+2)\cdot(3n+5)+\ell_u$ and $w(\ell_u)=(2\cdot \ell_u+2)\cdot(3n+7m)+\ell_u$, respectively. If a drawing with $\ell_u=1$ exists, i.e., if the $3$-SAT instance is satisfiable, the smallest drawing will have width and height $h(1) = 12n+21$ and $w(1)=12n+28m+1$, respectively. On the other hand, there is always a drawing with $\ell_u=2$ as this allows to set $\ell(x)=\ell(\neg x)=3$ which trivially fulfills the constraints of each clause gadget. Hence, if the $3$-SAT instance is unsatisfiable, we obtain 
a smallest drawing of width and height $h(2) = 18n+32$ and $w(2)=18n+42m+2$, respectively. Assume now for a contradiction that there was a $9/4$-approximation for \OctC problem. Given a positive instance of $3$-SAT, this approximation would yield a drawing  $\Gamma$ with area at most \begin{equation}
\frac{9}{4} h(1) \cdot w(1) = \frac{3}{2} h(1) \cdot \frac{3}{2}w(1) = (18n+31.5) \cdot (18n+42m+1.5) < h(2) \cdot w(2),    \notag
\end{equation}
i.e., a drawing that is \emph{smaller} than any drawing where $\ell_u=2$. Thus, in $\Gamma$, we have that $\ell_u=1$, i.e.,  we can actually decide in polynomial time whether there is a drawing with $\ell_u=1$, a contradiction unless P=NP.$\hfill\qed$
\end{proof}

\begin{figure}[p]
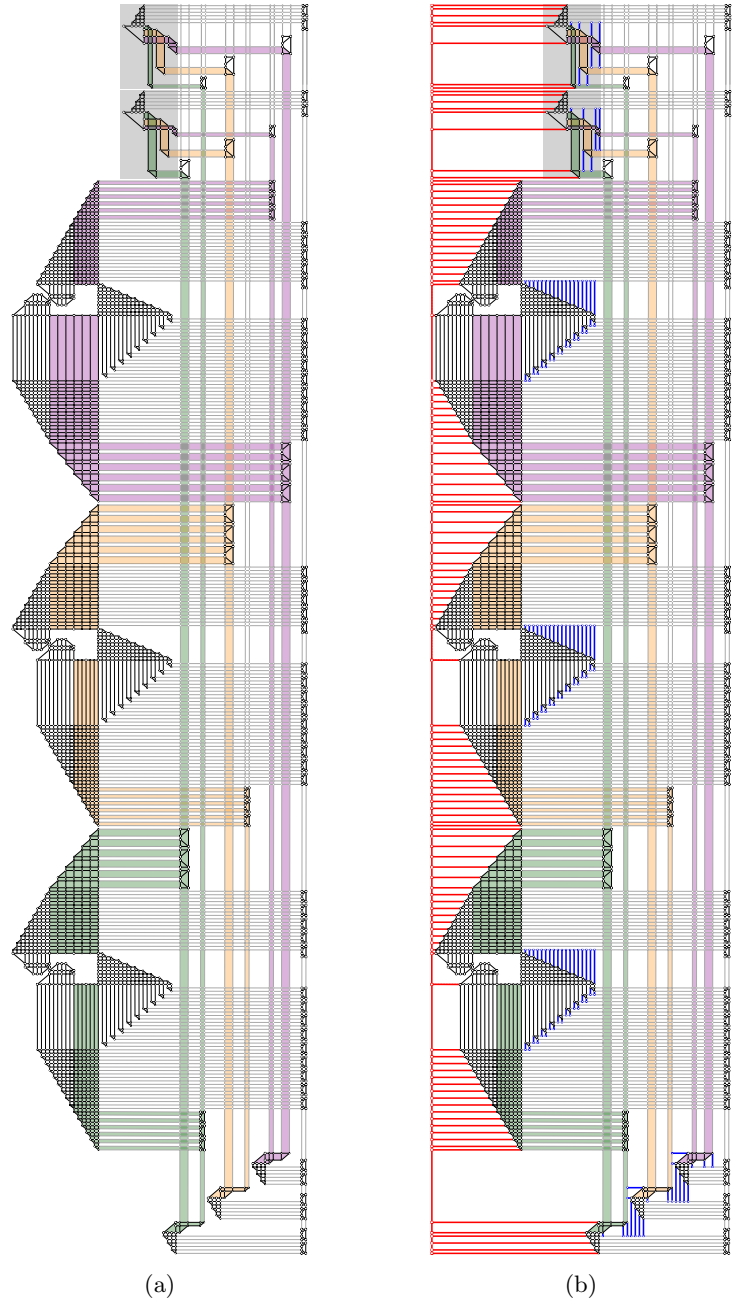

    \centering
     \begin{subfigure}{0.45\textwidth}
         \centering         \includegraphics[height=0.863\textheight,page=15]{octilinear-np.pdf}
        \caption{}
        \label{fig:parameterized-np-hardness-large-a}
    \end{subfigure}
    \begin{subfigure}{0.45\textwidth}
         \centering         \includegraphics[height=0.863\textheight,page=14]{octilinear-np.pdf}
        \caption{}
        \label{fig:parameterized-np-hardness-large-b}
    \end{subfigure}
    \caption{(a) Example of the reduction descibed in~\cite{DBLP:journals/algorithmica/BekosFK19} with formula $(x_1 \lor x_2 \lor  x_3) \land (\neg x_1 \lor x_2 \lor \neg x_3)$ in the generalized case where $\ell_u\ge 1$. The is assignment $x_1= x_2=\top, x_3=\bot$. (b) Modified version of the construction as described in Section \ref{sec:paranp}.}
     \label{fig:parameterized-np-hardness-large}
\end{figure}

\section{Omitted Details from Section~\ref{sec:paranp}}
\label{app:paranp}
\paraNP*
\begin{proof} We consider the NP-hardness construction by Bekos et al.~\cite{DBLP:journals/algorithmica/BekosFK19} for the more general case where $\ell_u\ge 1$. Refer to Fig.~\ref{fig:parameterized-np-hardness-large}. In particular, Fig.~\ref{fig:parameterized-np-hardness-large-a} is the construction of~\cite{DBLP:journals/algorithmica/BekosFK19}, while Fig.~\ref{fig:parameterized-np-hardness-large-b} is the figure depicting the construction we describe here. See also Figure~\ref{fig:octilinear-np-app} for enlarged figures. The thicker edges and the square vertices are the ones added in the procedure described as follow (in particular, blue edges are added in Part 1 and red edges are added in Part 2).

\smallskip \noindent
\textbf{Part 1: $\varphi=1$.} We have to insert additional segments to deal with reflex corners introduced by the variable gadgets for variables $x_i$ with $i \geq 2$ and the rerouting gadgets to obtain $\varphi=1$. We do so as follows. First, observe that each variable gadget for variable $x_i$ (with $i \geq 2$) adds two reflex corners. We subdivide the first edge of the literal path of literal $\neg x_{i-1}$ above twice and connect the variable gadget's reflex corners to the created dummy vertices by vertical segments. For each reflex corner introduced by rerouting gadgets, we observe that there is a vertical edge directly to the left or to the right. We subdivide this horizontally visible edge and connect the reflex corner to the newly created dummy vertex by a horizontal segment.
Observe that the only remaining reflex corners belong to parity gadgets. As discussed in our recap of the NP-hardness construction of Bekos et al.~\cite{DBLP:journals/algorithmica/BekosFK19}, the reflex corners of all parity gadgets belong to the outer face, thus the outer face is the only face that is bounded by a non-convex polygon. 

\smallskip \noindent
\textbf{Part 2: $\kappa=8$.} It remains to augment our construction for $\varphi=1$ to obtain the result for $\kappa=8$. In the process, $\varphi$ will become non-constant. To do so, we add another dummy vertex above each reflex corner on the outer face except for those being part of a block of a parity gadget. We then connect each reflex corner with the corresponding dummy vertex with a vertical segment. Finally, we connect all dummy vertices with a horizontal path. It is straight-forward to see that the only remaining reflex corners occur in parity gadgets and in the outer face. Moreover, each parity gadget is part of a separate face. Thus, we conclude that each face contains at most $\kappa=8$ reflex corners, which concludes the proof.$\hfill\qed$
\end{proof}


\begin{figure}[h]
    \centering
    \begin{subfigure}[b]{\textwidth}
     \centering
        \includegraphics[page=18,scale=0.65]{octilinear-np.pdf}
        \caption{Variable gadgets and parity gadget of $x_1$.}
        \label{fig:octilinear-np-app-a}
    \end{subfigure}
    \hfil
    \begin{subfigure}[b]{\textwidth}
     \centering
        \includegraphics[page=17,scale=0.65]{octilinear-np.pdf}
        \caption{Variable gadgets and parity gadget of $x_1$.}
        \label{fig:octilinear-np-app-b}
    \end{subfigure}   
\end{figure}
\begin{figure}[p]
\ContinuedFloat
    \centering
    \begin{subfigure}[b]{\textwidth}
     \centering
        \includegraphics[page=20,scale=0.65]{octilinear-np.pdf}
        \caption{Parity gadget of $x_2$.}
        \label{fig:octilinear-np-app-c}
    \end{subfigure}
    \hfil
    \begin{subfigure}[b]{\textwidth}
     \centering
        \includegraphics[page=19,scale=0.65]{octilinear-np.pdf}
        \caption{Parity gadget of $x_2$.}
        \label{fig:octilinear-np-app-d}
    \end{subfigure}   
\end{figure}
\begin{figure}[p]
\ContinuedFloat
    \centering
    \begin{subfigure}[b]{\textwidth}
     \centering
        \includegraphics[page=22,scale=0.65]{octilinear-np.pdf}
        \caption{Parity gadget of $x_3$.}
        \label{fig:octilinear-np-app-e}
    \end{subfigure}
    \hfil
    \begin{subfigure}[b]{\textwidth}
     \centering
        \includegraphics[page=21,scale=0.65]{octilinear-np.pdf}
        \caption{Parity gadget of $x_3$.}
        \label{fig:octilinear-np-app-f}
    \end{subfigure}   
\end{figure}
\begin{figure}[p]
\ContinuedFloat
    \centering
    \begin{subfigure}[b]{\textwidth}
     \centering
        \includegraphics[page=24,scale=0.65]{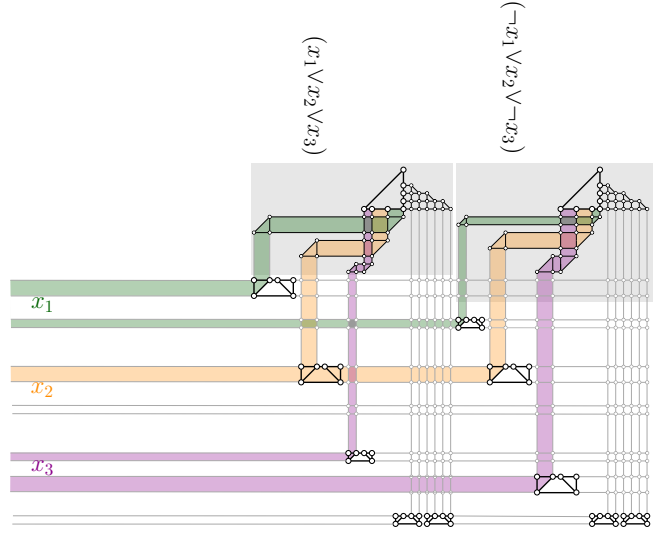}
        \caption{Clause gadgets.}
        \label{fig:octilinear-np-app-g}
    \end{subfigure}
    \hfil
    \begin{subfigure}[b]{\textwidth}
     \centering
        \includegraphics[page=23,scale=0.65]{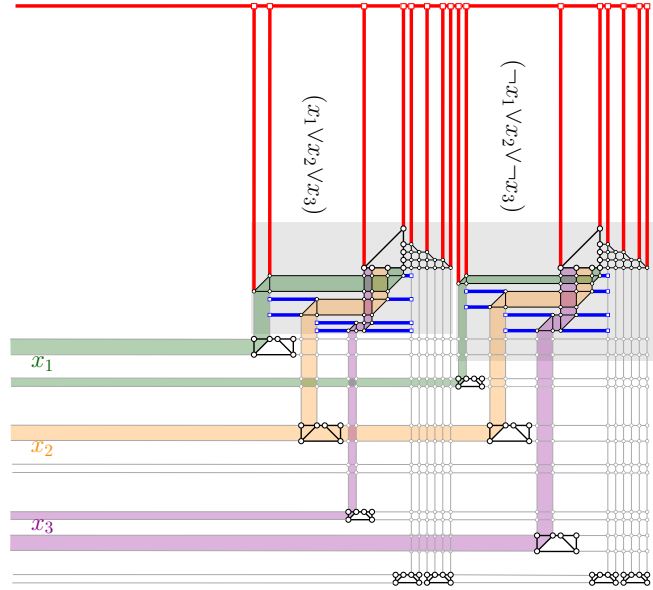}
        \caption{Clause gadgets.}
    \end{subfigure}   
    \caption{Example of the reduction with formula $(x_1 \lor x_2 \lor  x_3) \land (\neg x_1 \lor x_2 \lor \neg x_3)$ and assignment $x_1= x_2=\top, x_3=\bot$. (a), (c), (e), and (g) are the graph used to proof indescibed in~\cite{DBLP:journals/algorithmica/BekosFK19}, while (b), (d), (f), and (h) are the graph constructed for the proof of Theorem~\ref{thm:paraNP}.}
    \label{fig:octilinear-np-app}
\end{figure}

\section{Omitted Details from Section~\ref{sec:fpt}}
\label{app:fpt}

\fewext*
\begin{proof}
Every vertex forming a reflex angle can be aligned with an edge or a vertex incident to the same face. By Lemma~\ref{le:facesbounded} it follows that every such vertex can be aligned to $O(\omega)$ elements, which can be either vertices or edges. After choosing the alignments, we also have to order the subdivision vertices that occur on the same edge. For each of the at most $\omega$ subdivided edges, there are at most $\omega$ subdivision vertices added. This gives rise to at most $\omega! =2^{O(\omega^2)}$ permutations. The lemma follows.$\hfill\qed$
\end{proof}

\extrealiz*
\begin{proof}
The $\Leftarrow$ direction is immediate, as given a realization of $\mathcal{R}_+$ we can remove edges and smooth vertices. For the other direction, suppose that we have a realization $\Gamma\in \mathcal{R}'$. Every vertex at a reflex angle in an internal face is horizontally aligned with some edges incident to the same face. Concerning the outer face, if we add a drawing $\Gamma_J$ of $J$ so that $\Gamma$ is in the internal face of $\Gamma_J$ and such that all the internal angles are $\frac{\pi}{2}$, we have that every reflex angle in the external cycle of $\Gamma$ is horizontally aligned either to another edge of the same cycle or to an edge of $J$. We can add such horizontal edges, eventually subdividing an edge, and then we obtain a realization of an extension of $\mathcal{R}''$.$\hfill\qed$
\end{proof}

\end{document}